%% file: papDepReg27Apr12.tex
\documentclass[12pt]{article}

\input{deff.tex}

\begin{document}

\noindent{\Large \bf  Sequences of regressions and their dependences\\[6mm]}
{\large  NANNY WERMUTH\\[4mm]}
{\it Department of Mathematics, Chalmers Technical University, University of Gothenburg, Sweden,
and International Agency of Research on Cancer, Lyon, France\\[-2mm]}

{\small 
\em \noindent{\bf ABSTRACT}: In this paper,  we define and study the concept of traceable regressions. These are sequences of regressions in joint or single responses for which a corresponding regression graph captures
not only an independence structure but represents, in addition, conditional dependences that
permit the tracing of pathways of dependence.  We give the properties needed  for transforming 
these graphs and graphical criteria to   decide whether a path in the graph
induces a dependence.  The much stronger constraints on distributions that are faithful 
to a graph are compared to those needed for traceable regressions.   \\[1mm]
}

\section{Introduction and motivation}
{\bf  \nn \nn Single and joint response regressions.} Sequences of  regressions are arguably the most important statistical tool  
in observational and  interventional studies for  investigating pathways  of dependences and  hence  development 
over time.  In each regression, one distinguishes  {\bf \em response} variables and {\bf \em regressor} variables; with  responses depending on the regressors.

In  applications,  the substantive context determines which variable pairs  are modeled by a  conditional  independence and which are taken to  be dependent because they are needed in a generating 
process of the joint distribution. Suppose one regressor is a risk factor for a response, then quite different sizes of dependence strength will be relevant if  this response  is the occurrence of a common cold, or the infection with an  HIV virus or an accident in a nuclear plant, since the prevention of these risks  is  judged to be of quite different importance. 

There may be single or joint responses, where only the latter permit to model simultaenously occurring effects of an intervention.  Components of  joint responses may be  discrete or  continuous random variables or be mixed of both types.  Typically,  a subset of variables  is taken as given,
possibly determined by   study  design,  and its components are named {\bf \em context variables} since they describe the context or background or the basic features of  individuals under study.

 The generated joint density factorizes into  an ordered sequence of conditional densities  of  the responses, which we call shortly regressions,   and into  a  joint marginal density  of the context variables.   Under mild conditions, estimation of sequences of regressions can be decomposed into separate tasks  for each response component of the factorization, using   well-developed tools such as linear or logistic regressions or conditional Gaussian regressions, which permit  joint responses to be mixed of discrete and continuous component variables; see Lauritzen and Wermuth (1989), Edwards and Lauritzen (2001).
Tailored to the requirements in many specific situations,  special results are available  to estimate the form and parameters of univariate  and joint conditional distributions. 

However, many consequences of sequences of regressions can already be derived if one does not know  or estimate the  involved parameters  but just  uses an associated graph and   properties of  graph transformations. Relevant, important results concerning independences  in sequences of regressions have been obtained only recently;  see  Sadeghi and Lauritzen (2012) and Wermuth and Sadeghi (2012).  The additional properties needed to draw conclusions about induced dependences are set out in this paper.

Sequences of regressions are an essential part of  longitudinal studies, named also cohort or panel studies
in medical,  economic and social science research. Prominent examples are  the Framingham heart study, the European Community household panel or the Swiss HIV cohort study.
 By using  regression graphs,  it will become possible to simplify analyses and interpretations of  sequences of regressions
 and to   directly  compare dependences arising in different types of sequences of regressions for the same set of variables, or  in sequences of regressions   for subsets of variables  studied for subpopulations. The results in this paper prepare for these possibilities in applications.\\[-4mm]

 {\bf Independences and dependences given by regression graphs.}  Sequences of univariate, that is of  single-response regressions,  have been represented by {\bf \em directed acyclic graphs}.  With regression graphs,  directed acyclic graphs are extended  by including two types of undirected graph, one for joint responses, the other for joint context variables.
  Nodes of the graph 
 represent random variables.  Distinct node pairs are  coupled  by at most one edge so that a regression  graph is one type of what  in graph theory are called {\bf \em simple graphs}.   Each missing edge of a regression graph
 corresponds to a conditional independence where the conditioning set depends on the type and position of  the missing edge, the graph is therefore also one type of {\bf independence graph}.
%To derive the {\bf \em independence structure} of  a regression graph,  that is    all independence statements implied by the graph, one requires properties  that are common to all probability distributions and two additional ones that have been named the intersection property and the composition property; see Sadeghi and Lauritzen (2012).

Properties or axioms  for combining independence statements have been studied  by  Dawid (1979) and Pearl  (1988).  Their
 connections to graphs have been  discussed and modified  in information theory; see Studen\'y (2005) and
 Ln\v{e}ni\v{c}ka and Mat\'u\v{s} (2007).   Different types of extensions have been proposed  in the computer science literature;  see Castillo et al. (1997),  
 Flesch and Lucas (2007). %,  but without relating them to possible generating processes of joint distributions. 
 But, for instance, by requiring a property called  strong transitivity, one excludes even  the whole family of regular joint Gaussian distributions.
 By contrast, this family forms  a subclass of what we introduce here as  traceable regressions.

  The  {\bf \em independence structure} of  a  graph is the set  all independence statements implied by the graph.   These are well-studied
for    regression graphs,  but  with
 important results   obtained only recently. For instance,    a proof by Sadeghi and Lauritzen (2012)  implies equivalence of a  {\bf \em pairwise Markov property},
 that is of  the set of independences  attached to the  missing  edges of a given regression graph,   to the  {\bf \em global Markov property}, the criterion known to  give all independence statements implied by the graph.  For two regression graphs with  identical node sets and with the same set of coupled node pairs but with different types of edge,  there is  a   simple graphical criterion to decide  
 whether the two graphs define nevertheless the same independence structure, that is whether they are  {\bf \em Markov equivalent}; see Wermuth and Sadeghi (2012).\\[-3mm]
 
   %The name was chosen to honor Andrej A. Markov (1856-- 1922)
% for introducing the notion of conditional independence to simplify seemingly complex distributions. In particular, 
%the composition and the intersection property are shown to be necessary and sufficient for this equivalence.
 
 %In this paper,   composition, intersection and weak-transitivity are derived as properties of the edge matrix calculus for transforming graphs.
 
{\bf  Tracing pathways of dependence.} 
Much less is known about the dependence structures that can be captured by  graphs. Since  graphs do not distinguish between additive and interactive effects
 of  regressor variables on  responses, nor between linear and nonlinear types of dependences, it has
 been argued by Wermuth and Lauritzen (1989) that graphs  may represent {\bf \em research hypotheses about dependent variable  pairs needed   to generate the  joint  distribution}.  For this, each edge present in the graph 
indicates  a conditional dependence, where the conditioning set depends on the type and position of  the edge present,
while  the  form of the dependence is not specified. %We say then that the graph is a {\bf \em edge-minimal} for the generated distribution.

 For tracing pathways of dependences, dependence-inducing sequences of edges  of  different type  are the focus of interest, while independences just  lead to simplified strengthened interpretations of the relevant dependences. 
 In this paper, we set out the properties of traceable regressions and show, in particular that these
 properties impose mild constraints on the types of generated distribution in contrasts with 
 serious constraints  required  in general  for    faithful distributions. This notion was
 introduced   by Spirtes, Scheines and Glymour (1993)  for distributions in which all 
independence statements hold that are implied by a graph and no others.

Tracing pathways of dependence goes back to the geneticist Sewall Wright (1889--1988), who introduced it in 1923 as {\bf \em path analysis}  for
sequences of  univariate linear regressions.  He suggested to judge the goodness-of-fit  of a research hypothesis,  represented by 
a directed  acyclic graph, by comparing  observed correlations   with those that are expected if  the data had been generated over the graph. His  rules for computing   expected marginal correlations, trace all pathways that induce
a dependence by marginalizing.

%   so that  no variable is taken to be explanatory for itself. %Path analysis was  discussed and generalized in the statistical literature   by Tukey (1954), Goodman (1974)  and Wermuth (1980), introduced into sociological research by Duncan (1966) and suggested  as a tool for causal analyses by Spirtes et al. (1997).   

 The extension of tracing  pathways of dependences, when there is conditioning on variables in addition to marginalizing, became  
feasible  after  a first {\bf \em separation 
criterion} had been formulated by Pearl (1988) and proven  by Geiger, Pearl and Verma (1990)  to give the global Markov property of directed acyclic graphs. When separation fails then there is at least  one path in the directed acyclic graph that may induce a dependence
 by marginalizing over one subset of variables and  conditioning on another set. Here, such a  path is said to be edge-inducing since it leads to a transformed  graph.
 \\[-4mm]

{\bf Structure of the paper.}
In Section 2, we  introduce and discuss dependence base regression graphs and traceable regressions.
Section 3 contains examples of  tracing paths and  of  planning  future follow-up studies  on the same topic so that there are no  paths distorting a generating dependence of interest. % based on decisions which dependences can remain undistorted. 
 Small  Gaussian families of distributions are used to illustrate   independence properties of traceable regressions.  In Section 4, several discrete   families of distributions  are given to  show  how the properties of traceable regressions can be violated. %and which additional property turns traceable regressions into faithful distributions.
 In Section 5,   the known properties  of an edge matrix calculus 
 to transform  graphs are collected first. These are  used  to derive new properties of  transforming regression graphs and  to distinguish traceable regressions from
 distributions that are faithful to regression graphs. %to discuss computational aspects and relations to more traditional graph theory concepts.
A short discussion  ends the paper.\\[-9mm]

%  \Greg is one type of {\bf independence graph} since  to every missing edge there is at least one  independence statement  iassigned.

\section{Definitions and  terminology}

{\bf \nn \nn Some  terminology for graphs.} Most of the following definitions are standard or evocative and   listed for completeness. 
A graph consist of a  {\bf \em node set} ${\bm N}=\{1,\dots d_N\}$ and of  {\bf \em edges} that couple node pairs.
In simple graphs,   edges couple exclusively distinct node pairs by at most one edge so that  the endpoints $i$ and $k$ of an  ${\bm{ik}}$-{\bf \em edge} never coincide.  For an $ik$-arrow,  $i\fla k$,  node $k$ is commonly named the {\bf \em parent} of node 
$ i$. % In Figure \ref{regpain}, $U$ and $V$ are parents of $X_b$, both provide explanations for high depression scores before treatment, $X_b$. % and the set of parents of $i$ is denoted by $\bm{\pa_i}$.  
 % In a {\bf \em a split of  set $N$}, two  subsets of $N$ partition $N$.  
 
    For a regression graph, {\bm \Greg},  there is an ordered partitioning  of the node set as $N=(u,v)$ where  $u$ contains the response nodes, each having possibly several parent nodes and $v$ contains context nodes, none of which has a parent node; see for instance Figure 1 below. There are  {\bf \em  three types of edge sets},  $\bm{E_{\flas}}$ for directed dependences of  responses on their regressors, $\bm{E_{\dals}}$  for undirected dependences among components of 
 a  joint response, and $\bm{E_{\fuls}}$ for undirected dependences among context variables.          
 
  An {$\bm{ ik}$}{\bf \em-path} connects the     {\bf \em path endpoints} $i$ and $k$ by a sequence of edges. An $ik$-path can be an edge, otherwise it has  {\bf \em distinct inner nodes} such that  each edge visits  an inner node once.   There is is an {$\bm a${\bf \em -line path}, if all its inner nodes are in subset $a$ of $N$. 
  A path  of arrows is {\em direction-preserving} if all  its arrows  point  in the same direction. %There   is  a {\bf \em  chordless path}    if  for none of the path's uncoupled node pairs $(h,k)$,   there is an  $hk$-edge present in the graph.

For $a,b$ arbitrary disjoint  subsets of $N$, one says there is  a {\bf \em  path between} $\bm b$ {\bf \em and} $\bm a$ if   one endpoint   is in $a$ and the other endpoint  is in $b$, while we say  there is  a {\bf \em  path from} $\bm b$ {\bf \em to} $\bm a$ if  the subsets are ordered as $(a,b)$ that is if direction-preserving paths may start in $b$ and point to $a$, but not vice-versa.   In a  direction-preserving $ik$-path, node $k$ is  named an {\bf \em ancestor} of $ i$ and node $i$ a {\bf \em descendent} of $k$.

A {\bf \em subgraph, induced by  a subset} $ \bm a$ of the node set  $N$, consists of  the nodes within $a$ and of the edges present in the graph within $a$.
A special type of induced subgraph,  needed here, consists of three nodes and two edges. It is named a  {\bf{ \sf  V}}-{\bf \em configuration} or just a
{\sf V}. Thus, a three-node path forms a $\sf V$  if its induced subgraph has  two edges.

In    a  {\bf \em complete graph}, every node pair is coupled by an edge. In a {\bf \em connected subgraph},
every node can be reached by a path.  The {\bf \em connected component} of a regression graph are the disjoint  connected graphs that remain when all its arrows are removed. %; see the example in Figure 1 below.  
 Nodes within a connected component are said to be {\bf \em concurrent nodes}.
 \\[-4mm]

{\bf  Generating sequences of regressions and graphs.} 
We consider   $d_N$  random variables $Y_i$, which may be discrete or continuous or a mixture of both types. For a more formal definition of the measure spaces, asked for by a referee, see for instance Lauritzen  and Wermuth (1989). The variables   have  labels  in  node set $N$  and form  a vector variable, denoted by  $Y_N$.  In the following,  an element $i$ of $N$ is not distinguished from the singleton $\{i\}$ and the union sign for combining subsets 
of $N$ is often omitted.

For  $i,k$ a node pair and $c\subset N\setminus\{i,k\}$, we write  $\bm {i \ci k|c}$ for
  $Y_i, Y_k$  conditionally independent
 given $Y_c$. If such an independence constraint is  satisfied by a density $f_{ikc}$, 
$$  i \ci k|c \iff  (f_{i|kc}=f_{i|c}) \iff f_{ik|c}=( f_{i|c}f_{k|c}). $$

It has become common to say that a joint  family of densities $f_N$
can be {\bf \em generated over a chain graph} if it factorizes according to  a set ordering of the nodes, called a chain, and $f_N$ satisfies all independences
implied by the graph. Different types of chain graph and corresponding models for discrete variables are discussed by Drton (2009).

When independence structures are the focus of interest, one starts traditionally with the graph. Regression graphs  in node set $N$ have  three types of edge sets,
$E_{\flas}$, $E_{\dals}$, and $E_{\fuls}.$  For a regression graph, denoted by \Greg,
 there is  a split of the node set  as $N=(u,v)$,  so that concurrent response nodes  are  in $u$ and  concurrent context nodes in $v$, 
sets of ordered concurrent nodes are denoted by $g_j$ for $j=1, \ldots, J$.  Subgraphs induced by  any $g_j$ are undirected. 
Undirected subgraphs induced by  $g_j$ within $v$ have edges $i\ful k$ and are  commonly called  {\bf  \em concentration graphs}. Undirected 
  subgraphs induced  by $g_j$ within $u$ have  edges $i\dal k$ and are called {\bf \em covariance graphs}.
  
  For $g_j$ in $u$,  nodes in  
$g_{>j}=g_{j+1}\cup g_{j+2}, \ldots, \cup g_J$ are said to be in the {\bf \em past of $\bm{g_j$}}.
  Arrows  may start from any node, except from those in  $g_1$, but never point to a node in $g_{>j}$. 
With $g_{>J}=\emptyset$, 
the  basic factorization of  $f_N$ generated over a regression graph is 
 \begin{equation} f_N= f_{u|v}f_v  \text{ with  }   f_{u|v}=\txt \prod _{j\in u} f_{g_j|g_{>j}}   \text{ and }  f_{v}=\txt \prod _{j \in v} f_{g_j} .
\label{factdens} \end{equation}
Several orderings of $g_j$ may give the same factorization as explained below for Figure 1.

Here,  tracing of pathways is of main interest, hence we start instead with  a {\bf stepwise generating  process of $\bm{f_N}$} for which    $N=(u,v)$ and    {\bf \em one   ordering  of $\bm{g_j} $ is fixed}.  In this process, the density of variables in   $g_{J}$ is generated first, the one of  $g_{J-1}$ given $g_{J}$ next, up to the density of $g_1$ given $g_{>1}$. Then,  variable pairs   needed to generate $f_N$ give the edge set of \Greg with Definition 1 below and   the factorization of equation \eqref{factdens} results.

 For a {\bf \em variable pair
$Y_i, Y_k$ needed in the generating process of $f_N$}, we say it is conditionally dependent  given $Y_c$   for some $c\subset N\setminus\{i,k\}$ and write
 $\bm{i \pitchfork k|c}$ and  a graph is {\bf \em edge-minimal}  for a distribution generated over it, if every missing edge in the graph corresponds to a conditional independence statement and every edge present to a dependence. A family of densities $f_N$  generated over an edge-minimal graph   changes if any  one edge is removed from the graph.  

 \noindent \begin{defn} {\bf Defining pairwise dependences of   \Greg}. An edge-minimal  regression graph specifies  with   $g_1<\dots<g_J$ a generating process for $f_N$, where   the dependences 
 \begin{eqnarray} \label{pairw}
 i \dal k: \n  i\pitchfork k | g_{>j} \nn \nn \nn  \n &\,&  \text{for } i, k \text{ concurrent response nodes in }  g_j \text{ of } u, \nonumber\\
 i \fla k: \n  i \pitchfork k| g_{>j}\setminus \{k\}  &\,&  \text{for  response node } i \text{ in } g_j  \text{ of  }  u  
 \text{ and  node } k \text{  in } g_{>j}, \nn \nn\\
i\ful k:\n     i\pitchfork k| v \setminus\{i,k\}  &\,& \text {for }  i, k \text{ concurrent context nodes in }  g_{j} \text{ of  } v ,\nonumber
 \end{eqnarray}
define the   {\bf edges present} in \Greg.  The meaning of each {\bf edge missing}  in \Greg results with the dependence sign  $\pitchfork$  replaced by  the independence sign $\ci$.
\end{defn}

Thus,   for the given  order of the components $g_j$,  the graph implies for  each variable pair $i,k$  either  conditional dependence or independence given the same  type of conditioning set, with   $i \dal k$ for  two response
nodes, with $i\fla k$ for  $i$ a response node in $g_j$ and $k$ a node in the  past of $g_j$, with $i\ful k$ for  two context nodes.   Notice that except  for context nodes,   each pair of  variables  is  exclusively conditioned on variables that are in the past of  the $g_j$ that contains node $i$. This permits to model simulateanously occurring  effects of an intervention while   this is not possible if the graph is  directed acyclic or if it is another type of  chain graph.

 % or, as mentioned before,  to represent a research hypothesis on the variables needed to generate  $f_N$;
%see Wermuth and Lauritzen (1990).
%  %Furthermore, for the given ordering of the components $g_j$,
% every variable pair $Y_i, Y_k$ is either conditionally dependent or independent given the same conditioning set,
%as  specified in \eqref{pairw}.

%An $ik$-path  is  a {\bf \em cycle} if its endpoints coincide, that is if $i=k$.  
% A  {\bf \em chordal graph} is an undirected graph without  chordless cycles in four or  more nodes. For such a cycle,
 %the subgraph induced by every consecutive three  nodes  forms  a $\sf V$ in the graph.

   Different generating processes may lead to the same regression graph and hence to the same  implied independence structure. Then,
   some  components,  $g_j, g_{j+1}, \ldots, g_t$, say of \Greg,  have an   interchangeable labeling because  they  induce  disconnected undirected subgraphs.   Such  components are displayed in Figure \ref{hypfig} within  stacked boxes.  
   In a connected \Greg, connected  stacked  components  $g_j, \ldots, g_t$  have the nodes in $g_{>t}$ as  their {\bf \em common past}  and   nodes in    $g_{<j}=g_{1}\cup g_{2}, \dots , \cup  g_{j-1}$ as their {\bf \em common future}.  
    For a generating process,   one of the possibly many  {\bf \em compatible orderings} is fixed. In each,  arrows  point to response nodes in the  common future but never  to a node in the common past.

    In Figure \ref{hypfig} below, $g_6$ and $g_7$ are in $v$, all other connected components  are in $u$.
    The order implied by the arrows in  $ E_{\flas}$ of \Greg 
 remains  unchanged if, for instance,  the two disconnected subgraphs induced by $g_3$ and $g_4$ are  interchanged  or if  they are replaced by  a  single dashed-line complete graph in nodes
$ \{4,5,6,7\}$.

  Recall that connected components of \Greg  are uniquely  obtained as the connected subgraphs that remain after  deleting all arrows from the regression graph and  keeping the undirected  edges and all nodes. Thus, for any given graph, it is not necessary to show stacked boxes, but they  are sometimes  included to reflect the first ordering, the  prior knowledge about possibly joint responses
 and joint context variables. 
   By convention, we number nodes  and  components $g_j$ of \Greg    first from top to bottom,  then from left to right.   
  In Figure \ref{hypfig}, $g_3=\{4,5,6\} $ and $g_8=\{12,13,14\}$ contain three nodes, each of $g_2$ and $g_6$ contain 
   two nodes, all others contain a single node.

  \begin{figure} [H]
\centering
 \includegraphics[scale=.46]{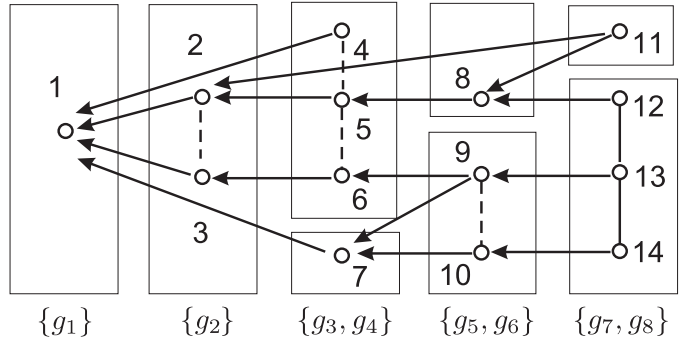}
 \caption[]{\small{A  regression graph in   14 nodes and  node set partitioned into 8 connected components; single responses in $g_1,  g_4, g_5$ and joint responses in $g_2, g_3, g_6$; context variables  in $g_7, g_8$.}}
  \label{hypfig} \vspace{2mm}
 \end{figure}

   Single responses correspond in the statistical model to univariate regressions, joint responses to multivariate regressions, including the seemingly unrelated  regressions of Zellner (1962).  In Figure \ref{hypfig},  seemingly unrelated regressions belong to the subgraphs induced by each of the three node sets $\{2,3,5,6\},  \{5,6,8,9\},
   \{9,10,13,14\}.$
      \\[-4mm]

{\bf General and special properties of probability distributions.}  For $i,h,k$ single, distinct  indices and $a,b,c,d$ disjoint subsets of  index set $N$, where only $d$ may be empty, there are the common independence properties $(i)$ to $(iv)$ which are satisfied by  all probability distributions.  The discussed  properties  $(v)$ to $(viii)$ constrain distributions, but they  permit the use of just the graph
to derive different types of consequences for families of distributions $f_N$ generated over \Greg.
%instead of knowledge about the form and the types of parameters  in $f_N$. 
 \begin{eqnarray*}
 (i)&\!  &\text{\bf symmetry: }  a\ci b|c \iff  b\ci a|c, \\
(ii) &\, &\text{\bf  contraction: } (a\ci b|cd   \text{ and } b\ci c|d)  \iff ac\ci b|d, \\
(iii) &\,&  \text{\bf  decomposition: } a\ci bc|d     \implies (a\ci b|d \text{ and } a\ci c|d), \\
(iv) &\, & \text{\bf weak union: }  a\ci bc|d     \implies (a\ci b|cd \text{ and } a\ci c|bd).
\end{eqnarray*}
Joint distributions, for which the reverse implications of $(iii)$ and of $(iv)$ hold, have as additional properties,
respectively, 
 \begin{eqnarray*} 
 (v) &\, & \text{\bf composition: }  (a\ci b|d \text{ and } a\ci c|d) \implies  a\ci bc|d,\\
(vi) &\, &  \text{\bf   intersection: }   (a\ci b|cd \text{ and } a\ci c|bd)  \implies a\ci bc|d.
\end{eqnarray*}
Properties $(v)$ and $(vi)$ are needed to derive the independence structure implied by \Greg.
Two further types of properties are to be considered for tracing pathways of dependence,  
\begin{eqnarray*} 
 (vii) &\, & \text{\bf set transitivity: }  (a\ci b|d \text { and }  a \ci b|cd) \implies  (a\ci c|d \text{ or }  b\ci  c|d) \, , or  \\[-13mm]
\end{eqnarray*}
\begin{eqnarray*} 
 (viii) &\, \!\!& \text{\bf singleton transitivity: } 
 (i\ci k|d   \text{ and } i\ci k| hd) \implies (i \ci h| d \text{ or } k\ci h | d).
 \end{eqnarray*}
Thus, distributions that satisfy set transitivity are also singleton-transitive, since  $c$ may contain only one element.
Singleton transitivity requires for a conditional independence of $Y_i, Y_k$ given $ Y_d$ and given $Y_h, Y_d$ to hold both, there has
to be at least one additional independence given $ Y_c$ involving  $Y_h$,  the additional variable  in the conditioning set.  
 It is 
unfortunate that, in the literature, the  term weak transitivity has sometimes been used for property  $(vii)$ and sometimes for $(viii)$.

We shall show that 
set transitivity, $(vii)$,  is used  in addition to $(i)$ to $(vi)$ in  transformations of \Greg by which no edge of the starting graph
 gets removed and by which an edge criterion for the global Markov property is obtained, while only singleton transitivity, $(viii)$,  is needed  in addition to $(i)$ to $(vi)$ to  decide with a given edge-minimal \Greg, whether a path is  inducing a dependence for its path endpoints or not.

Singleton transitivity is  a feature of what we define  below as traceable regressions. So far,  it had been  known to be common to all positive  binary distributions where,  for instance, for  ($1\pitchfork 2$ and $1\pitchfork 3$) either $2\ci 3$ can hold or $2\ci 3|1$ but  not both; see  Simpson (1951). It also holds in all regular Gaussian distributions; 
see for instance Studen\'y (2005),  Corollary 2.5 in Section 2.3.6. 

On the other hand, set transitivity imposes stronger constraints on any specific  distribution; see  for instance the discussion of Figure 1 for trivariate binary  distributions in Wermuth, Marchetti and Cox (2009)% but it permits conclusions for families of distributions by ignoring particular    parametric  constellations. 
 It also 
excludes some  regular Gaussian families of distribution such as the following.\\[-4mm]

{\bf A regular Gaussian family violating set transitivity.}  For $N=(u,v)$, let $Y_u$ and $Y_v$ be mean-centered vector variables  with  a joint Gaussian distribution. Let them  have equal dimension, $dim$, the components of
$Y_u$  and of $Y_u$  be mutually independent and   all elements in the covariance matrix $\cov(Y_u, Y_v)=\Sigma_{uv}$ be nonzero,  then every component of $Y_u$ is  dependent on  every component in $Y_v$  and  
$$    \cov(Y_u)=\Sigma_{uu} \text{ diagonal},  \nn \nn   \cov(Y_v)=\Sigma_{vv}\text{ diagonal}.$$ 
Let further  the components of $Y_v$ have   equal variances $\omega>1$ and the  equal variances of the components $Y_u$ be $\kappa>\omega+1$.
Whenever in the described situation $\Sigma_{uv}$ is orthogonal,  the joint covariance matrix is invertible so that the joint distribution is regular and the  marginal independences carry over to conditional independences  so that also 
$$    \cov(Y_u|Y_v)=\Sigma_{uu|v} \text{ diagonal},  \nn   \cov(Y_v|Y_u)=\Sigma_{vv|u} \text{ diagonal}.$$
Set transitivity is always violated, for a split $v=(a,b)$, $c=\{1, \ldots, dim\}$ and $d=\emptyset$. 
 This family extends the example  in equation (8) of Cox and Wermuth (1993).\\[-4mm]

 {\bf  Some important properties of {\bm \Greg} and $\bm{f_N}$.}   Two basic types of {\sf V}s in \Greg need to be distinguished.
There are  {\bf \em collision {\sf V}s}:  
$$ i \dal \snode \fla k,  \nn \nn  i  \fra \snode \fla k, \nn \nn i\dal \snode \dal k,$$
and {\bf \em transmitting {\sf V}s}: 
$$ i \fla \snode \fla k, \nn \nn  i \fla \snode \ful k, \nn \nn  i \ful \snode \ful k, \nn \nn i \fla \snode \fra  k, \nn \nn 
  i \fla \snode \dal k\,.$$
  
 Recall that two different graphs in the same node set are Markov equivalent if they define nevertheless the same independence structure, that is the  set of all independences implied by the graph. The  skeleton of a graph consists of its nodes and its set of edges, irrespective of the type of edge. It  results by replacing each edge present by a full line.
\begin{lem} {\bf Markov equivalence}. {\rm (Wermuth and Sadeghi, 2012).}
Two  regression   graphs with the same skeleton are Markov equivalent if and only if their  sets of  collision {\sf V}s
are identical.
\end{lem}
  
A more compact  characterization of the pairwise independences  in Definition 1 is based on the  notion of
anterior paths. Recall first that with $N=(u,v)$, there are only undirected full-line paths within $v$ and  there are  only arrows pointing from  $v$   to $u$. 
An {\bf \em anterior $\bm{ik}$-path} is either a descendant-ancestor $ik$-path, or  a context nodes $ik$-path, or a    descendant-ancestor $iq$-path with a  context-nodes  $qk$-path  attached to it, 
 \begin{displaymath}
i \
\fla \underbrace{\overbrace{ \snode \fla \snode,  \ldots, \snode \fla  q} ^{\text {\normalsize{ancestors of  \textit{i}}} }\ful \snode, \ldots , \snode \ful k}_{\text{\normalsize{anteriors of \textit{i}}}}.\\
 \end{displaymath}
We denote  the joint set of anteriors of  nodes  $i,k$ by ${\rm ant}_{ik}=\{{\rm ant}_{i}\cup {\rm ant}_{k}\}\setminus \{i,k\}.$
 Similarly, for any subset $c$ of $N$, the anterior set  of nodes within $c$ is denoted by  ${\rm ant}_{c}$. 
  
 The intersection $(vi)$ and the composition property  $(v)$ are needed for Lemmas 2 and 3.  By using them,   the independences  attached to the missing edges  of \Greg in Definition 1 reduce to the more compact  statements $i\ci k| {\rm ant}_{ik}$   and this leads to the  definition of an active path in \Greg due to Sadeghi (2009) for a more general class of graphs.
 
Let   $\{a,b,c,m\}$ partition  $N$, where $c$ denotes a conditioning set of interest for $a,b$ and $m$ the set of nodes to be ignored that is to be marginalized over. Only $c$, $m$ or both may be empty sets. 
 A {\bf \em path  in {\bm \Greg} is active given $\bm c$} if of its inner nodes, every  collision node is in $c\cup {\rm ant}_{c}$ and every transmitting node is in $m$. For graph transformation, such  a path is also said to be {\bf \em edge-inducing}.

\begin{lem} {\bf Global Markov property of {\bm \Greg}.} {\rm (Sadeghi, 2009).}  The regression graph \Greg implies $a\ci b|c$ if and only if there is no active path  in \Greg between $a$ and $b$ given $c$. \end{lem}

\noindent \begin{lem} {\bf Equivalence of the pairwise and the global Markov property}. {\rm (Sadeghi and Lauritzen, 2012).} The independence  structure of  \Greg is equivalently defined by its lists of the three types of   missing edges  and by  its global Markov property.
\end{lem}

  To make {\sf V}s dependence inducing, we take  an edge-minimal regression graph for $f_N$, assume   properties $(i)$ to $(vi)$  and, in addition property $(viii)$, that is singleton  transitivity. We then say \Greg is a  {\bf \em dependence base} for $f_N$ since the implications 
of this type of  graph can be derived with respect to both  independences and dependences. We note first that by enumeration in Definition 1,  the inner node of each collision {\sf V}  is excluded from the defining conditioning set for its endpoints, while  the inner node of each transmitting {\sf V} is 
included in it. This observation is generalized with Lemma 4.
 
 \begin{lem} The conditioning set of any independence statement implied  by \Greg for  the endpoints of any of its {\sf V}s, 
  includes the inner node if it is a transmitting {\sf V}  and excludes the  inner node if it is collision  {\sf V} . \end{lem} 
 
 % A family of densities  generated over an edge-minimal  \Greg factorizes as in equation \eqref{factdens}.
 \begin{proof}  %Then the  graph implies dependence of $i$,o and $k$,o for any conditioning  set.  but may imply  
 The statement results directly with Lemma 2. \end{proof}

%Reformulating singleton transitivity  shows when   {\bf \sf V}s of \Greg are dependence inducing.% in $f_N$ generated over it.
Let now a  {\sf V}  in  a dependence base \Greg have  endpoints   $i,k$ and   inner node $\rm o$. Then by Definition 1 und Lemma 4, there is  at least one  $c$  with $c \subseteq  N\setminus \{i,k,\rm o\}$ such that $ i\ci k|c$ is implied if  $(i, {\rm o}, k)$ is a collision  {\sf V} and  $ i\ci k|{\rm o}c$   if  $(i, {\rm o}, k)$ is a transmitting  {\sf V}.
\begin{prop} {\bf Dependence inducing {\bf \sf V}s}.  For  $(i, {\rm o}, k)$ any  {\sf V} of a dependence base \Greg and each $c \subseteq  N\setminus \{i,k,\rm o\}$  
such that this regression graph implies one of   $i\ci k|c$ or  $i\ci k|{\rm oc}$, the following two equivalent statements hold:  \begin{eqnarray*}
 &-& (i, {\rm o}, k) \text{ forms a collision {\sf V} }\iff (i\ci k|c \implies   i \pitchfork k|{\rm o} c),   \\
 &-& (i, {\rm o}, k)   \text{ forms a  transmitting  {\sf V}} \iff (i \ci k|{\rm o} c \implies   i \pitchfork k|c) \,  .
\end{eqnarray*}
\end{prop}

 \begin{proof} For $c=\emptyset$,  collision  {\sf V}s are Markov equivalent and   transmitting {\sf V}s are Markov equivalent  by Lemma 1.
    By edge-minimality, both  edges of any  {\sf V}  indicate  conditional dependence for  pairs $i$,o and $k$,o  and by Definition 1, $i\ci k$ holds for an inner collision node and $i\ci k|$o for an inner transmitting node.   Including the inner node of a collision {\sf V} into the conditioning set, or  excluding the inner node of a transmitting {\sf V} from the conditioning set, generates an active path  by Lemma 2.  Such a path induces a dependence unless  singleton transitivity is violated which contradicts an assumption. Similarly for  $c\neq \emptyset$, an independence 
 is  implied by \Greg  if  there is no active path between $i$ and $k$ given $c$ by Lemma 4, but  an active path is generated just as for $c=\emptyset $.\end{proof}

 We can now define sequences of regressions that permit the tracing of pathways  of dependence  for $f_N$ when  $a,b,c,d$ denote  disjoint subsets of $N$ and only $d$ may be  empty. % or $i,{\rm o},k$ denoting the distinct nodes of a {\sf V} in \Greg for  $d=N\setminus\{i, {\rm o},k\}$.
 \pagebreak
 \begin{defn}  {\bf Traceable regressions.}  We say $f_N$   results  from traceable regressions if \begin{enumerate}
 \item it could have been generated over a dependence base regression graph, \Greg,
  \item it has three equivalent decompositions of the  joint independence  $b\ci ac|d$ 
$$
(i) (b\ci a|cd   \text{ and } b\ci c|d) ,  \nn 
(ii)     (b\ci a|d \text{ and } b\ci c|d) , \nn 
(iii)  (b\ci a|cd \text{ and } b\ci c|ad)  ,
$$\n \\[-16mm]
\item dependence-inducing  {\sf V}'s  of  \Greg are also  dependence-inducing for $f_N$.  \end{enumerate}
 \end{defn}

 Decompositions $(i)$ to $(iii)$ in Definition 2 combine the previously discussed   properties $(ii)$ to  $(vi)$.
Symmetry of independences, that is property $(i)$,  holds trivially as in all probability distributions. Undirected edges
 correspond to symmetric dependence statements. For each arrow  $i\fla k$ in \Greg, symmetry of dependence  holds only in the following weak sense. From Definition 1 for  $i$  in  $g_j$,   there is some $c \subseteq g_{>j}\setminus k$ with $f_{i|kc}\neq f_{i|c}$ used in the generating process.  Then, for $Y_k$   regressed instead on $Y_i, Y_c$, also  $f_{k|ic}\neq f_{k|c}$.  
 
Notice that traceable regression behave like regular Gaussian families  generated over an edge-minimal \Greg.
  Therefore,  
  for traceable regressions, a violation of set transitivity can occur only when there are at least two paths connecting the same node pair; see  the family of regular Gaussian distributions given above that violates set  transitivity and   for further   examples   Wermuth and Cox (1998).   We call these special types of parametric constellations  {\bf \em path cancellations} as they result for a pair $i,k$  after combining dependences  induced by active $ik$-paths  in  such a way that the joint contributions of all paths cancel.  \\[-9mm]

\section{Applications and illustrations of traceable regressions} 
%The following   discussion is to illustrate  different   types of  goals when  tracing paths.
 %We assume that the generated  distribution satisfies the properties needed for successfully tracing paths. 

{\bf Tracing paths}. Whenever  a pathway of dependence is traced  in terms of a graph, one uses implicitly
that every edge present is a strong enough dependence   to be of interest in the given substantive context
and that every {\sf V} along a path is dependence-inducing for its endpoints, since otherwise, no dependence 
is implied for the path endpoints.

Figure \ref{regpain} shows  a well-fitting regression  graph for nine features observed for  patients. The regression   graph  represents a research hypothesis on  the sets of regressors needed for each response  to  generate  the joint distribution. 
In this  example, we use
 data  of Kappesser (1997)  on 201   chronic pain patients, where 
  variable descriptions and detailed statistical analyses are given  in  Wermuth and Sadeghi (2012), not in this paper.
  
  The graph does not  contain any information on the types of the dependence, but 
supplemented by estimates for the dependences,  one can use the graph to interpret pathways of dependences.  
For instance the path $Y, Z_a,A,B$ leads, together with the parameters estimated with linear and logistic models, to the following interpretations. 

    \begin{figure}[H] 
\centering 
\includegraphics[scale=.54]{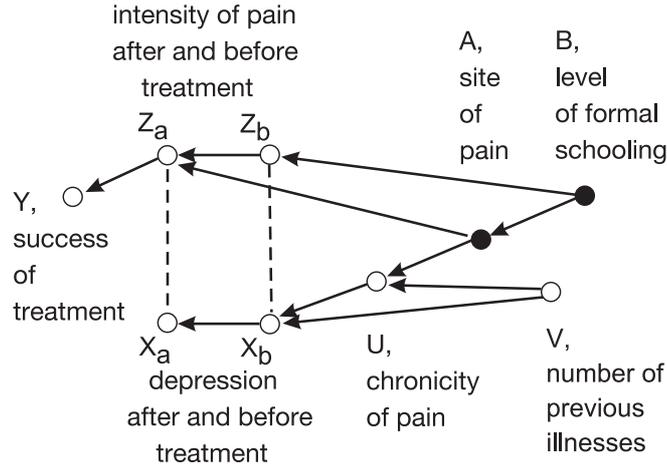}
\caption{\small Regression graph, well compatible with the data and resulting from  statistical analyses.  Binary variables
are indicated  by dots, variables treated as continuous by circles.}  \label{regpain}
 \vspace{-2mm}
\end{figure}

Patients with a higher level of formal schooling are more likely to have  head or neck pain
than back pain. For  patients with head or neck pain, the intensity of pain is better reduced after treatment  than for the back pain patients. %Success of treatment. $Y$,  is judged by the patients three months afterleaving the pain clinic.  
  For lower pain intensity scores after treatment,    treatment  is the more successful the lower the pain intensity. For higher pain intensity scores  after treatment,  there are no systematic changes in $Y$.

% If this nonlinear dependence is overlooked,
%because one uses for instance a search strategy that permits only  incorporation of  main effects, then depression after treatment, $X_a$, and site of pain, $A$, are falsely included as important regressors for  generating $Y$.

The graphs in Figure \ref{transfpaingraph} are consequences of the generating graph in Figure \ref{regpain}.
 Figure \ref{transfpaingraph}a) implies that site of pain, $A$, would  show  a direct effect on $Y$ if  the two symptoms of chronic pain  before and after treatment were  either not measured or just omitted from the list of potentially important regressors. 
Similarly, chronicity of pain, $U$, would show a  direct effect on $Y$ if,  in addition, site of pain, $A$, is  omitted in \ref{transfpaingraph}b).  

\begin{figure}[H] 
\centering
\includegraphics[scale=.52]{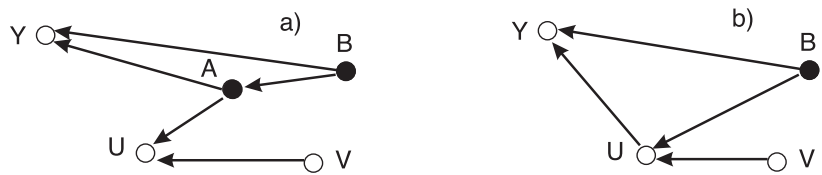} 
\caption{\small The graph of Figure \ref{regpain}  transformed,  preserving the original ordering for the  remaining variables  by a) marginalizing over symptoms before and after treatment, $X_{a}$, $Z_{a}$,  $Z_{b}$, $X_{b}$; b)  marginalizing over symptoms before and after treatment and, in addition, over site of pain, $A$.} \label{transfpaingraph}
\vspace{-3mm}
\end{figure}

To derive and interpret transformed graphs well, such  as those  in Figure \ref{transfpaingraph}, and  more complex  graphs involving both  marginalizing and conditioning, one needs to know the general properties  
of transforming  regression graphs and realize that in general,  induced dependences may not be reflected in significant statistical test results, in 
particular  for small sample sizes or  weak dependences attached  to edges in the generating graph.
 \\[-4mm] % of corresponding  generated distributions that permit the tracing of pathways of dependences  and that we therefore name traceable regressions. %and  of distributions that satisfy all and only those independences implied by the starting graph. 
%As we shall see, the conditions for a distribution to be faithful to a regression graph are  considerably stronger 
%than to be a  traceable regression.

{\bf Planning future follow-up studies}. To show how tracing of  active paths  may lead to an improved planning of follow-up studies, we use the   generating process, represented by the graph in Figure \ref{jamlar97}, adapted from  Robins and Wasserman (1997), and assume that all those dependences  are strong that
correspond to edges present in the graph.
\begin{figure}[h]\label{jamlar97}
\begin{center}
\includegraphics[scale=.55]{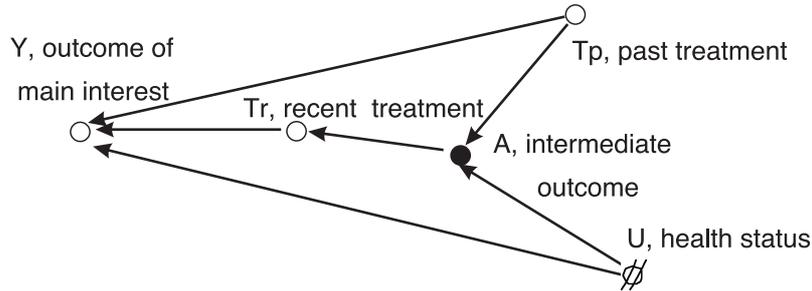}
\caption{\small Generating process in five
variables, missing edge for   $(T_{\rm p}, U)$   due to full randomized allocation of individuals to treatments, and missing edges for $(T_{\rm r}, U)$ and $(T_{\rm r}, T_{\rm p})$ due to randomization conditionally on $A$; $U$ expected to be unobserved in a follow-up study.}
\end{center}
\vspace{-5mm}
\end{figure}

Suppose that in the planned   study, it will be possible to observe all  variables  of Figure 4 except for $U$, because the tools needed to 
diagnose the health status, $U$,  before treatment will not be available. Marginalizing over $U$ is  indicated  in Figure 4 by a crossed out node, $\margn$.
Then  $U$ is excluded from any conditioning set for $Y$, the main response of interest. In general, whenever  no active path is generated, one may proceed safely with  estimating  an  effect, a dependence of  main interest, in the follow-up study.

With $U$ unobserved,  the dependence of  $Y$ on the past  treatment $T_p$ will  always be modified, since  by excluding also the intermediate outcome, $A$, and recent treatment, $T_r$
from the list of regressors, one generates  the active path $Y,T_r, A, T_p$, while by including either $T_r$ or $A$ or both
as  regressors for $Y$, one generates the active path $Y, U, A, T_p$; see Lemma 2.  The former is an example of an overall effect deviating from a conditional effect and the latter is an {\bf \em example of indirect confounding}.

If on the other hand, the dependences of $Y$ on the recent treatment, $T_r$, is of main interest, then $T_p$ is a common ancestor and the
path $Y, T_p, A, T_r$ becomes active by marginalizing over the inner nodes; an {\bf \em example of direct confounding}.
But no active path is  generated between $Y$ and $T_r$ when $A$ and  $T_p$ are  regressors in addition to $T_r$, so that the conditional dependence of $Y$ on $T_r$ given $A,T_p$ can be estimated.

Even though it may in principle  be possible to recover the generating dependence given some distributional assumptions; see e.g. Wermuth and Cox (2008), one needs to obtain  very precise  estimates  to make  any correction worthwhile since  poorly estimated parameters may also lead to bad corrections.

%Using the summary graph, on which  dependences are still estimable without distortions when compared to the 
%dependences in a  generating process with more variables.

Both types of confounding can also  be detected using graphical criteria  on transformed graphs in reduced node sets, named summary graphs; see Wermuth (2011). For constructing summary graphs by removing repeatedly  single nodes, one needs to  take  into account that any given node can be a collision node on one path and a transition node on another path.  This contrasts with the  graph transformations in this paper, where  different types of active paths are  closed  in sequence.
 \\[-4mm]%In Proposition 4 below, an  edge matrix criterion is given for induced dependences. These  calculations can be carried out to count the number of paths inducing dependences for any pair $(i,k)$. In particular, if there is only one such path, direct and indirect confounding will be absent.  

 {\bf Examples of small Gaussian regression graph models.}  
We  illustrate next the intersection and the composition property  by describing  two different types of complete regression graphs in three nodes and the associated  saturated models  in the special case of  regular families of  Gaussian distributions for variables standardized to have zero mean and unit variance. Parameters are attached to the edges of the graphs.
Example I shows that the intersection property is implicitly used with backward selections  of important regressors in  multiple regressions  and Example II how the composition property is relevant for  selecting  important regressors in multivariate regressions.\\[-4mm]

\noindent {\bf  Example I:  a  complete single response graph with  two context variables.} The following complete graph in nodes $1,2,3$ \begin{center}
\includegraphics[scale=0.48]{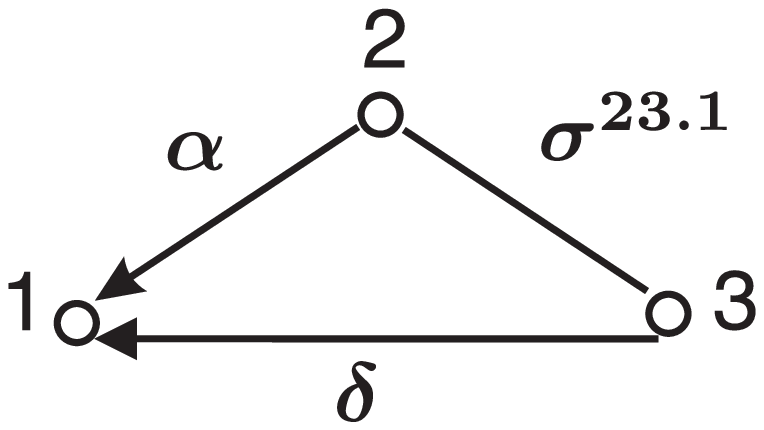}
\end{center}
defines implicitly for  standardized Gaussian variables,  $Y_1,Y_2, Y_3$  three nonzero parameters measuring dependence in
$${ \rm E}( Y_1| Y_2,  Y_3)=\alpha  Y_2 + \delta  Y_3   \nn \nn 
{ \rm E}( Y_2 Y_3)=\rho_{23}\nn \nn \sigma^{23.1}=-\rho_{23}/(1-\rho^2_{23}) \, ,
 $$ 
 where $\rho_{23}$ denotes the marginal correlation  of $Y_2,Y_3$ and $\sigma^{23.1}$ the concentration in their bivariate distribution, that is after marginalizing  over $Y_1$.  For this complete graph,  $\alpha \neq 0$  means $1 \pitchfork 2|3$,  $\delta \neq 0$ means  $1 \pitchfork  3|2$,  and  $\sigma^{23.1}\neq 0$ means $ 2
\pitchfork 3$. With $\alpha=\delta=0$, one requires $1\ci 2|3$ and $1\ci 3|2$ and  removes the 12-edge and the 13-edge from the complete graph so that node 1 remains isolated from  $2 \ful 3$. 
For the resulting   graph, the seemingly obvious interpretation  $1\ci (2,3)$
 requires the intersection property.\\

\noindent{\bf  Example II: a complete joint response graph with a single regressor.} The following complete graph
\begin{center}
\includegraphics[scale=0.48]{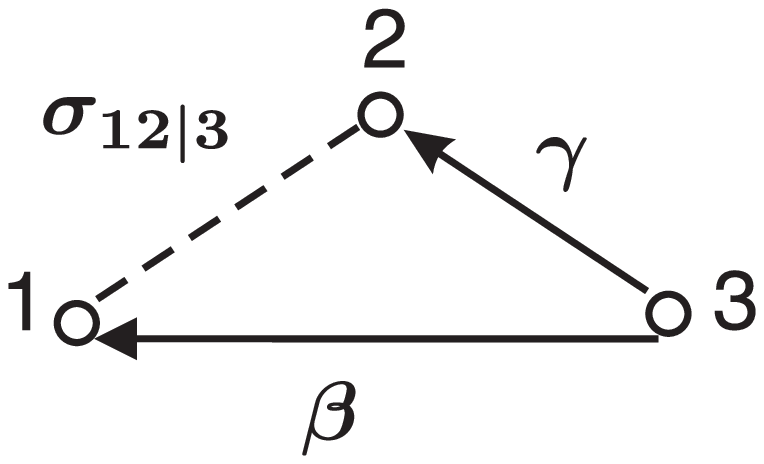}\vspace{-2mm}
\end{center}
defines for  standardized Gaussian variables three non-vanishing parameters, $\beta, \gamma, \sigma_{12|3}$, in
$$  
{ \rm E}(Y_1| Y_3)=\beta   Y_3\nn \nn
{ \rm E}( Y_2|Y_3)=\gamma  Y_3  \nn \nn { \cov}( Y_1 Y_2|Y_3)=\sigma_{ 12|3}\, .
$$ 
Here,  $ \sigma_{12|3} \neq 0$ means  $1  \pitchfork 2| 3,$  $\beta \neq 0$ means   $1 \pitchfork 3,$ and $  \gamma \neq 0$ means $  2 \pitchfork  3.$  With $\beta=\gamma=0$,  one requires $1\ci 3$ and $2\ci 3$ and  removes the 13-edge and the 23-edge from the complete graph so that node 3 remains isolated from  $1 \dal 2$. 
For the resulting   graph, the interpretation  $(1,2)\ci 3$
requires the composition property. \\[-4mm]

%The composition property holds in all regular Gausssian distributions and in traceable regressions,
%The four
%properties that are common to  all probability distributions  are illustrated for completeness next.

{\bf Standard properties for combining independences.}  Properties $(ii)$ to $(iv)$  that are common to all probability distributions with a given density, are illustrated next by using
 the    directed acyclic graphs in the three ordered nodes
$(1,2,3)$ shown in Figure \ref{inc3dags}, again for standardized Gaussian distributions. \\[-4mm]% and parameters  attached to the edges of the graphs.

\noindent{\bf Example III:  a complete directed acyclic graph.} The complete graph in nodes $1,2,3$ of Figure \ref{inc3dags}a)
gives for  standardized Gaussian variables three nonzero parameters, $\alpha, \delta, \gamma$,  measuring dependence in
$${ \rm E}(Y_1| Y_2, Y_3)=\alpha  Y_2 + \delta  Y_3,  \nn \nn 
{ \rm E}(Y_2| Y_3)=\gamma  Y_3, \nn \nn  {\rm E}(Y_3)=0\,,
$$\
where $\alpha \neq 0$ means $1\pitchfork 2|3$,   $\delta \neq 0$ means $1\pitchfork 3|2$, and $\gamma \neq 0$ means $ 2
\pitchfork 3$. 

\begin{figure}[H]
\centering
\includegraphics[scale=0.43]{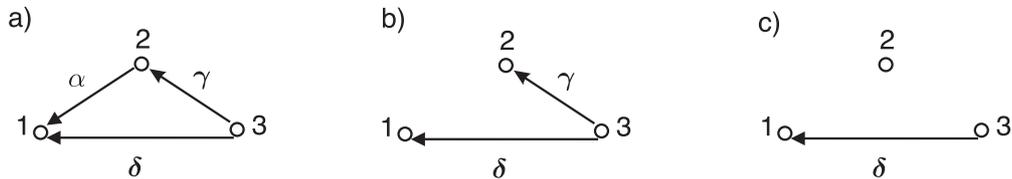}
 \caption[]{\small{Directed acyclic graphs in 3 nodes with parameters in  standardized Gaussian distributions attached to the edges; a) the complete graph, b) the graph implying $1\ci2|3$, c) the graph implying $2\ci (1,3)$.}}
  \label{inc3dags}
\end{figure}

The interpretation  of $\delta$ changes to $\delta \neq 0$ means $1 \pitchfork 3$ in Figure 5b) where $1\ci 2|3$ is implied by the graph. This reflects that a 
different family of distributions is generated when the $12$-edge is removed.
The graphs define implicitly the factorizations of $f_N$ in equation \eqref{factdens}, respectively, as
$$ f_{123}=f_{1|23}f_{2|3} f_{3}, \nn \nn  (f_{123}=f_{1|3}f_{2|3} f_{3}) \implies 1\ci 2|3, \nn \nn  (f_{123}=f_{1|3}f_{2} f_{3})\implies  2\ci (1,3)\,.$$ 

The  factorization of a joint density as specified with a complete directed acyclic graph is formally always possible. % though 
%it may not reflect the order of the generating process. 
Independence constraints imposed in sequence on two consecutive factors of $f_{123}$ generated as in Figure \ref{inc3dags}a),  such as $1\ci 2|3$ constraining $f_{1|23}=f_{1|3}$ changes the triangle in the graph of Figure \ref{inc3dags}a)  to a {\sf V} 
in Figure \ref{inc3dags}b)    and 
$2\ci 3$ constraining $f_{2|3}=f_2$  creates next an 
isolated node 2 and  $1\fla 3$,   in Figure \ref{inc3dags}c). 

The removal of the two arrows gives one direction of the contraction property, starting from the factorization to  Figure \ref{inc3dags}c) gives the other direction.
Given the factorization of any density to Figure \ref{inc3dags}c), marginalizing over $Y_3$ leaves $f_{12}=f_1f_2$
and marginalizing over $Y_1$ gives directly $f_{23}=f_2 f_3$ that is decomposition, while conditioning on
$Y_2,Y_3$ leaves directly $f_{1|23}=f_{1|3}$ and conditioning on $Y_1,Y_2$ gives $f_{3|12}=f_{3|1}$ that
is weak union.   Also in more complex situations, these three properties, $(ii)$, $(iii)$, $(iv)$, common to all probability distributions, can be derived  by transforming factorized densities.
\\[-9mm]

 \section{Violating  properties  of  traceable regressions.}
  Some small  discrete families of distribution  are given that are not traceable regressions. These may be extended and many similar
 families may be constructed.

 {\bf Violation of singleton transitivity.} As mentioned before, singleton transitivity is  satisfied in all regular Gaussian distributions and  in all  binary distributions.
 But the  following discrete family of distributions  for a $2\times 2\times 3$   contingency table 
 violates  singleton transitivity. It is  adapted from Birch (1963), equation (5.4).   We write $\pi_{ijk}$ for the joint probabilities
 of variables $A, B,C$ at levels  $i,j, k$  and e.g. $\pi_{+jk}=\txt{\sum_i}\pi_{ijk}$.
 Conditional probabilities e.g. for $A$ given $B, C$ are $\pi_{i|jk}=\pi_{ijk}/\pi_{+jk}$.\\[-6mm]
 %of total count $n=\text{constant} \times 4(1+\alpha+\alpha^2)$. 
% of single  variables.

\begin{table}[H]  
\caption{A family of distributions that violates singleton transitivity}\vspace{-3mm}
 \begin{center}
{\small $4 \pi_{ijk}(1+\alpha+\alpha^2)$, $\alpha>1$\\
 \begin{tabular}{l cc c cc c  cc  rrr  rr} 
 \cline{1-9}
&\multicolumn{2}{c}{$C: \nn k=1$}&&
\multicolumn{2}{c}{$k=2$}&&\multicolumn{2}{c}{$k=3$}&&& \\
\cline{2-3} \cline{5-6}   \cline{8-9} 
{$\mathbf  \sf  A /B:$}& $j=1$&$j=2$&&$j=1$ &
$j=2$&&$j=1$&$j=2$ &&&\\ 
\cline{1-9}  
$i=1$& $\alpha^2$ & $\alpha$ & &$\alpha$ & 1 &&1&$\alpha^2$ &&&\\
$i=2$ &$\alpha$ &1 && $\alpha^2$ & $\alpha$&&1&$\alpha^2$&&& \\
\cline{1-9} 
odds-ratio & \multicolumn{2}{c}{\fourl 1} && \multicolumn{2}{c}
{\fourl 1}&&\multicolumn{2}{c}{\fourl 1}&&& \\
 \cline{1-9}  
 \end{tabular}}
\end{center}
 \vspace{-6mm}
 \end{table}

Here,  the  conditional odds ratios being 1 imply that $A\ci B|C$ and the marginal probabilities of  $A,C$ and of $B, C$ 
show  that $A\pitchfork C$ and $B \pitchfork C$. Nevertheless, also $A\ci B$ since $$\txt{\sum_k}\, \pi_{i+k} \, \pi_{+jk}/\pi_{++k}=\pi_{i++}\pi_{+j+},$$
 a very special constellation discussed first by Darroch (1962) and  generalized 
by  Wermuth and Cox (2004), section 7, to general types of distributions that are also not dependence inducing.
Though one can construct families of distributions with such peculiar parametric constraints, it is difficult to imagine that they could capture a structure of interest in any substantive context when studying sequences of regressions.
 
 In a generating process of $f_N$,  singleton transitivity can be achieved when the individual regressions are
 permitted to vary independently of the other response components and of their common past.  This is reached,  in particular,  when the family to a complete graph has a rich enough parametrisation and only the independence  constraints of Definition 1 are imposed on \Greg. 
 \\[-4mm]

 {\bf Violation of the intersection property.} The intersection property is  always  satisfied by positive distributions  and in all in regular Gaussian distributions,  even though  the known necessary and sufficient conditions are less restrictive; see San Martin, Mouchart and Rolin (2005).
The following discrete family of distributions for a $2\times 2\times 3$ 
 contingency table violates    the intersection property. This happens whenever a pair of variables 
shares some common information. For three binary variables,  violation of the
 intersection property coincides with  the degenerate case of two variables being identical. \\[-6mm]
 \begin{table}[H]  
\caption{A family of distributions that violates the intersection property}\vspace{-3mm}
 \begin{center}
{\small $3\pi_{ijk}$, $0<\alpha\neq \beta<1,$ $2\alpha+\beta<1$\\\
 \begin{tabular}{l cc c cc c  cc  c cc  cc cc} 
 \cline{1-9}
&\multicolumn{2}{c}{$C: \nn k=1$}&&
\multicolumn{2}{c}{$k=2$}&&\multicolumn{2}{c}{$k=3$} \\
\cline{2-3} \cline{5-6}   \cline{8-9} 
{$\mathbf  \sf  A /B:$}& $j=1$&$j=2$&&$j=1$ &
$j=2$&&$j=1$&$j=2$ &\\ 
\cline{1-12} 
$i=1$& $\alpha$ & $0$ & &$\alpha$ & 0 &&0&$\beta$ &\\
$i=2$ &$1-\alpha$ &0 && $1-\alpha$ & 0&&0&$1-\beta$& \\
\cline{1-9} 
 \end{tabular}}
\end{center}
 \vspace{-6mm}
\end{table}

In the family shown in Table 2, $A\ci B|C$ and $A\ci C|B$, since 
$$\pi _{i|jk}=\pi_{i|k} \text{  and } \pi_{i|jk}=\pi_{i|j}$$
but $A\dep BC$. More precisely, $A\dep B$ since $\pi_{i|j}\neq \pi_{i}$ and
$A\dep C$ since $\pi_{i|k}\neq \pi_{i}$. The marginal joint distribution of $B,C$ shows
the type of common information shared by variables  $B$ and $C$. Variable
$B$ taking on level 1  coincides with $C$ taking on  value 1 or 2 and  $B$ being at level 2
coincides with $C$ being at level 3.

Thus, when the joint  distribution of $B,C$ had been generated by first knowing 
the distribution of variable $C$  and then generating the conditional distribution of $B$ given $C$, the levels of variable $B$ are  not  permitted to vary freely and thereby lead to the violation of  the intersection property.
 \\[-4mm]

 {\bf Violation of the composition property.}  The composition  property is  always  satisfied in regular Gaussian distributions and in multivariate symmetric binary distributions generated over directed acyclic graphs; see Wermuth, Marchetti and Cox (2009).
On the other hand, it is always violated when   pairwise independences do not imply mutual independence.

The following  binary family of distributions for a $2 \times  2 \times  2$ contingency table also violates the composition  property.
In this family, there is a  log-linear three-factor interaction since the conditional odd-ratios
for $A,B$ differ at the two levels of $C$.\\[-6mm] 
 \begin{table}[H]  
\caption{A family of distributions that violates the composition property}\vspace{-3mm}
  \begin{center}
{\small $8\pi_{ijk}$,  $0<2\alpha<1$\\\
 \begin{tabular}{l cc c cc c  cc  c cc  cc cc} 
 \cline{1-6}
&\multicolumn{2}{c}{$C: \nn k=1$}&&
\multicolumn{2}{c}{$k=2$} \\
\cline{2-3} \cline{5-6}    
{$\mathbf  \sf  A /B:$}& $j=1$&$j=2$&&$j=1$ &
$j=2$\\ 
\cline{1-12} 
$i=1$& $1+2\alpha$ & $1-2\alpha$ & &$1$ &$ 1$ \\
$i=2$ &$1-2\alpha$ & $1+2\alpha$ &&$1$& $1$  \\
\cline{1-6} 
odds-ratio & \multicolumn{2}{c}{$\{(1+2\alpha)/(1-2\alpha)\}^2$} && \multicolumn{2}{c}
{\fourl 1} \\
 \cline{1-6}  
 \end{tabular}}
\end{center}
 \vspace{-6mm}
\end{table}
More precisely, at level 2 of $C$, the variables $A,B$
are  independent while the dependence of this pair is strong at level 1 of $C$  whenever $\alpha$ is large.
At the same time, the marginal $AC$ and $BC$ tables reveal that $A\ci C$ and $B\ci  C$.

 Thus,  when regressing  the two components of a  joint response  $AB$ separately on $C$, one sees  no 
 separate effects, but the conditional dependence   of $A$ on $B$ changes  with the levels of $C$. 
 This  type of structure could in particular not be generated by a single  unobserved common explanatory variable or if all sets of variable with higher-order effects also have main effects in the regressions, that is  lowest order interactions.

% nor over a directed acyclic graph in three nodes that represents a dependence base;  see Marchetti and Wermuth  (2009), Lemma 1.
 
 With a pragmatic strategy  for model selection in which one  checks for higher order interactions only when there  are also main effects, one may  overlook  such structures that could  be of substantlve interest. % such as the  example from psychological research due to Lienert (1920--2001) and  reproduced in  Wermuth (1998).  
  For sequences of discrete joint responses, the violation will be detected when using the parametrization suggested by  Marchetti and Lupparelli (2011). 
 In general, the graphical  checks for nonlinearities and interactions, as proposed  by    Cox and Wermuth (1994),  provide some protection, but only for effects
 that are detectable also  in marginal trivariate distributions. \\[-9mm]

%In the case of larger structures, that imply independence for some subpopulations and dependence for others, as above,
%%it may be worthwhile to move to separate analyses and corresponding split graphs; see H\o sgaard (2003) and
%Wermuth and Cox (1998), Section 3.4 for an application.

\section{Transforming regression graphs} 

 The transformations of regression graphs to be introduced, are based on binary matrix  representations of \Greg.
 Our notation for these  edge matrices mimics  the one for parameter matrices in
Gaussian
sequences of regressions generated over the graph. There are   one-to-one correspondences between a zero in an edge  
matrix, a vanishing parameter in the regular Gaussian family of distributions and a conditional independence statement.\\[-4mm]

 {\bf  Linear 
sequences of regressions.} For a mean-centered  vector variable $Y_N$ with a regular Gaussian distribution generated over \Greg with a split $N=(u,v)$, the matrix of equation parameters, denoted by  $H_{NN}$, is  upper block-triangular and 
$$H_{NN}Y_N=\eta_N \text{ with }  W_{NN}=\cov(\eta_N) \text{ block-diagonal in the sizes of } g_j,$$
  where the submatrix of $H_{uu}$ in rows $g_j$ and columns $g_{>j}$ is $-\Pi_{g_j|g_{>j}}$, the negative of the population least-squares coefficient matrix obtained when regressing $Y_{g_j}$ on $Y_{>g_j}$.  The square diagonal submatrices in the sizes of $g_j$ are identity matrices. The submatrix $H_{vv}$  is  the marginal concentration matrix of $Y_v$, denoted by $\Sigma^{vv.u}$. This  implies $W_{vv}=\Sigma^{vv.u}$. The square  submatrices of $W_{uu}$ are
   $\Sigma_{g_jg_j|g_{>j}}$, the conditional covariance matrices of $Y_{g_j}$ given $Y_{>g_j}$. For  just two
   connected components $a,b$  the  parameter matrices are     
$$ H_{NN}=\begin{pmatrix} I_{aa} & -\Pi_{a|b.v} & -\Pi_{a|v.b}\\ 0_{ba} & I_{bb} & -\Pi_{b|v}\\
0_{va} & 0_{vb} &  \Sigma^{vv.ab} \end{pmatrix}\nn \nn  W_{NN}=\begin{pmatrix} \Sigma_{aa|bv}& 0_{ab}& 0_{av}\\ 0_{ba}& \Sigma_{bb|v}& 0_{bv}\\0_{va}& 0_{vb}& \Sigma^{vv.ab} \end{pmatrix},$$
where we use a Yule-Cochran notation for  $\Pi_{a|bv}$, the regression coefficient matrix of $Y_a$  regressed on $Y_b$, $Y_u$,  for instance $0_{ba}$ denotes a matrix of zeros, and $I_{bb}$ an identity matrix.

 For any split of  $N=(a,b)$,  to obtain  $f_{a|b}f_{b}$
we let $c=a\cap u$,   $d=b\cap u$, and get
 $$K_{NN}=\begin{pmatrix} H_{aa}^{-1} &-H_{aa}^{-1}H_{ab}  \\ H_{ba} H_{aa}^{-1}\;& H_{bb}-H_{ba}H_{aa}^{-1}H_{ab}
\end{pmatrix}\nn \nn  Q_{uu}=\begin{pmatrix} W_{cc}-W_{cd}W_{dd}^{-1}W_{dc}&W_{dd}^{-1}W_{dc}\\-W_{dd}^{-1}W_{dc}&
 W_{dd}^{-1}
 \end{pmatrix},$$
  by partial inversion of $H_{NN}$ with respect to $a$ and by partial inversion of $W_{uu}$ with respect to $b$;
 see for instance Marchetti and Wermuth (2009), Appendix 1.

\begin{lem} {\bf Orthogonalised linear equations.} \!\!{\rm (Wermuth and Cox (2004), Thm 1.)}  The Gaussian density 
$f_N=f_{u|v}f_v$ generated over \Greg is for any split $N=(a,b)$ transformed into $f_N=f_{a|b}f_b$
with $E(Y_a|Y_b)=\Pi_{a|b}$, $\cov(Y_a|Y_b)=\Sigma_{aa|b}$, ${\rm con}(Y_b)=\Sigma^{bb.a}$ as
  \begin{equation} \Pi_{a|b}=\In[K_{ab}+K_{aa}Q_{ab}K_{bb}], \label{ipab}\end{equation}
  \begin{equation} \Sigma_{aa|b}=\In[K_{aa}Q_{aa}\kcal_{aa}\T], \nn \nn \Sigma^{bb.a}=\In[H_{bb}\T Q_{bb}H_{bb}]\label{isym}.\end{equation}  \end{lem}

 \n \\[-4mm]
    {\bf  \nn \nn The edge matrices of regression graphs.} Edge matrices are binary matrix representations of graphs. They are symmetric for undirected graphs, upper block-triangular for arrows in a generating \Greg and upper-triangular for 
directed acyclic graphs. 
The essential change  compared to  the more traditionally used adjacency matrices is that ones are added along the diagonal of each square matrix. This has the effect that  sums of matrix products are  well-defined and can represent the closing of special types of path in graphs; such as in   equations \eqref{indpab} and  \eqref{indsym} below. \\[-4mm]

 Regression graphs have  three types of edge sets,
$E_{\flas}$, $E_{\dals}$, and $E_{\fuls}.$   The  edge matrix components of \Greg are 
 a $d_N \times d_N$ upper  block-triangular
matrix $\hcal_{NN}= (\hcal_{ik})$ such that
\begin{equation}
\label{hcal}
\hcal_{ik} =
\begin{cases}
 1 & \text{ if  and only if } i \fla  k \text{ or } i\ful  k \text{ in } G^{N}_{\rm reg} \text{ or } i = k, \\
 0 & \text{ otherwise,} \\
 \end{cases}
\end{equation}
and a $d_u \times d_u$ symmetric matrix $\wcal_{uu} = (\wcal_{ik})$ such that
\begin{equation}
\label{wcal}
\wcal_{ik} =
\begin{cases}
 1 & \text{ if  and only if } i \dal  k  \text{ in } G^{N}_{\rm reg} \text{ or } i = k, \\
 0 & \text{ otherwise,} \\
 \end{cases}\end{equation}
where, $E_{\dals}$ corresponds to  $\wcal_{uu}$, $E_{\fuls}$ to $\hcal_{vv}$, and
$E_{\flas}$ to $\hcal_{uN}$. 

Every regression graph \Greg can be represented by its edge matrices given in equations \eqref{hcal} and \eqref{wcal}.
Every dependence base \Greg defines in particular  a  corresponding  family of   Gaussian  regressions in which each edge present can be identified  by a single non-vanishing parameter, an off-diagonal element of $H_{NN}$ or $W_{u,u}$.\\[-4mm]

{\bf Partial  closure of paths.} 
Partial closure,  introduced by Wermuth, Wiedenbeck and Cox (2006),  is a matrix operator, denoted by $\zer_a(\cdot)$ which acts on row and collums  $a$ of a binary  matrix.  It is applied to  edge matrix representations of a  starting graph  in node set $N$ to give  the edge matrix representations of a new graph in which 
there is an additional $ik$-edge for a  pair $i,k$  that is in the starting graph  uncoupled but  connected by a specific type of edge-inducing
 $a$-line path.

 % for instance for  an undirected graph.% it is any $a$-line path, %for a  directed acyclic graph, it is a direction-preserving path of arrows.  

With partial closure, the set of nodes, node labels, and edges present in the starting graph,  are  preserved in the transformed graph so that the mappings  are graph  homomorphisms; for this notion see Hell and Ne\v{s}et\v{r}il  (2004), for corresponding
reparametrizations of exponential families see Wiedenbeck and Wermuth (2010).

\begin{lem} {\bf Basic properties of partial closure.} {\rm (Wermuth, Wiedenbeck and Cox, (2006)).}   Partial closure is $(i)$ commutative, $(ii)$ cannot be undone and $(iii)$ is exchangeable with selecting a submatrix. \end{lem}
By property $(i)$,  it is enough, for some purposes, to show how the operator acts on a single node.  By property $(ii)$, independences
can be removed but never reintroduced so that these transformations satisfy set transitivity. 
 Property $(iii)$ justifies  node and edge reductions since
closing edge-inducing  $a$-line paths in a large graph and then selecting a square  submatrix for a subset  containing $a$, gives the same result as selecting the square submatrix first and then closing  the  $a$-line paths.

   Because of property $(i)$, one can always permute the matrix $\fcal$  into $\tilde{\fcal}$ and start  partial closure with node $i$ corresponding to  position (1,1) of $\tilde{\fcal}$.  Then for   $b=N\setminus \{i\}$, 
  \begin{equation}%\setlength{\arraycolsep}{.5\arraycolsep}
 \label{def1zer}
\zer_i\,  {\tilde\fcal} =\In[\left(\begin{array}{rrr}
1 & \n \fcal_{ib}\\[2mm]
                \fcal_{bi}&\nn \fcal _{bb} +\fcal_{bi}\fcal_{ib}
              \end{array}\right)] , \end{equation}   
which says that  particular {\sf V}s in the graph are closed which  have node $i$ as inner node.
In the three small examples of Figure 6, an edge for node pair $1,3$ is induced with $i=2$.
\begin{figure}[H]
\centering
\includegraphics[scale=0.43]{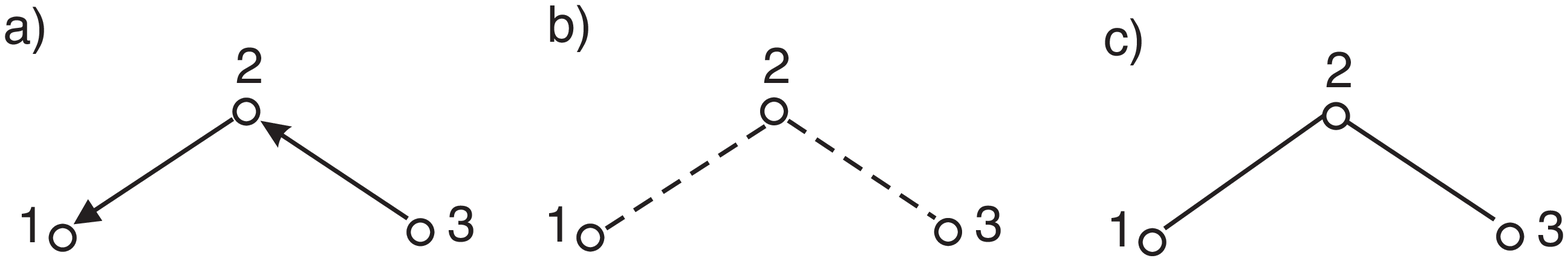}
 \caption[]{\small{Dependence base, 3-node graphs: a {\sf V} in a a) directed acyclic, b) covariance, c)  concentration graph;  an 
 active path (1,2,3) induces in a) $1\pitchfork 3$, in b) $1\pitchfork 3|2$, and in c) $1\pitchfork 3$}}
  \label{threeV}
\end{figure}

Applying $\zer_{i}$ to  the edge matrix of  a
directed acyclic graph, covariance graph or   concentration graph mimics, respectively, the  recursion relation for 
regression coefficients, covariances and concentrations; discussed  for instance  in Wermuth and Cox (1998). 

By
letting  the edge   induced by the three {\sf V} 's in Figure 6, `remember the type of edge at the path endpoints',
the induced edges become, respectively,
$$ \text{ a) } 1 \fla 3, \nn \nn \text{ b) } 1\dal 3, \nn \nn \text{ c) } 1 \ful 3.
$$

The transformation $\zer_a(\fcal)$  means that all $\sf V$s along $a$-line paths represented by  the edge matrix $\fcal$ are closed by an edge. The basic property $(i)$ implies  that the nodes in $a$ may be chosen for this  in any order. This requires   in particular that the inner nodes of the paths  of  $\fcal$ are of the same type,
either  all are collision nodes to form   {\bf \em collision paths},  or are all transmitting nodes. % and 
%(2) that induced edges can be consistently combined that is 
\begin{lem} {\bf Partial closure applied to {\bm  \Greg}}. The transformation $\kcal_{NN}=\zer_a( \hcal_{NN})$  
 closes  each $a$-line anterior path and $\qcal_{uu}=\zer_b(\wcal_{uu})$  each  dashed, $b$-line collision path.\end{lem}

\begin{proof} %All {\sf V}'s with inner nodes in $a$ are closed for paths that are represented by the edge matrix  acted on.
 Each   anterior path in \Greg and no other type of path is represented by  $ \hcal_{NN}$ and each dashed-line   path  in  \Greg and no other type of path is represented by $ \wcal_{uu}$. By Proposition  1,  a {\sf V} 
along the former is edge-inducing by marginalizing over its inner node and of the latter by conditioning on its inner node. 
 Remembering the 
 type of edge at the endpoints of each {\sf V} on an $a$-line path of $\hcal_{NN}$ leads to the same induced edge for the
 endpoints of the path irrespective of the order in choosing single  nodes of $a$.
 \end{proof}
 
% Notice that for a dependence base \Greg, each edge-inducing path of Lemma 6 is also dependence-inducing and that
% property 

%Lemma 6 assures for a partitioning $N=\{\alpha,\beta,c,m\}$ denoting  disjoint subsets of $N$ with at most $d=\emptyset$ that the  independence  $\beta\ci \alpha m|c $ implied by \Greg continues to hold after partial closure 

 {\bf  Closing active paths in regression graphs.} For   directed acyclic graphs,  it is known that the path criterion  on the starting  graph for 
 separation  of $\alpha$ from $\beta$ given $c$  can be reduced to an edge criterion after transforming  first the generating graph in terms of partial closure  and  closing  next the remaining   paths that are 
 relevant for deciding whether $\alpha \ci \beta| c$ is implied; see Marchetti and Wermuth (2009). This approach is now 
 extended to regression graphs and to  dependences in traceable regressions. For this, we take the  partitioning $N=\{\alpha, \beta, c, m\}$ of the node set of \Greg, $a=\alpha \cup m$,  $b=\beta \cup c$, and 
  $$\kcal_{NN}=\zer_a \hcal_{NN}, \nn   \qcal_{uu}=\zer_b \wcal_{uu}, \nn \qcal_{uv}=0, \nn \qcal_{vv}=\kcal_{vv}.$$ 
 
  \begin{prop}{\bf Induced edge matrices for $\bm{f_{a|b} f_b}$}. Sequences of regressions with graph  \Greg in node set $N=(u,v)$ and generating
 edge matrices $H_{NN}$ and $W_{uu}$ imply for $f_{a|b}f_{b}$, with the induced 
 regression graph $G^{N-a|b}_{\rm reg}$ for  $Y_a$ regressed on $Y_b$, as edge matrices 
 \begin{equation} \pcal_{a|b}=\In[\kcal_{ab}+\kcal_{aa}\qcal_{ab}\kcal_{bb}], \label{indpab}\end{equation}
  \begin{equation} \scal_{aa|b}=\In[\kcal_{aa}\qcal_{aa}\kcal_{aa}\T], \nn \nn \scal^{bb.a}=\In[\hcal_{bb}\T\qcal_{bb}\hcal_{bb}]\label{indsym}.\end{equation}
  \end{prop}
 
\begin{proof} Partial closure mimics  transformations of partial inversion such that all elements of the induced matrices are non-negative.
The zero entries in equations \eqref{ipab}, \eqref{isym} coincide with those in \eqref{indpab}, \eqref{indsym}, nonzero entries in the former correspond to ones in the latter; see Lemma 3 of Marchetti and Wermuth (2009) for more detail.
\end{proof}

Of the active paths, defined for Lemma 2 and   needed to decide for uncoupled  pairs $i,k$  of \Greg whether they are coupled in $G^{N-a|b}_{\rm reg}$, some remain uncoupled  after applying $zer_a \hcal_{NN}$ and $\zer_b \wcal_{uu}$ but get closed with the non-negative sums of edge matrix products in \eqref{indpab}, \eqref{indsym}.  Thus, as with partial closure, no edges get ever removed with the latter types of graph transformations so that set transitivity is used implicitly.

For the  $N=(a, b)$ as for Proposition 2,
let   ${\rm o}_a$ denote nodes  in $a$ and  ${\rm o}_b$ nodes in $b$. \begin{coro}  For   $i,k$ the endpoints of paths that are edge-inducing for $G^{N-a|b}_{\rm reg}$,  there are three types of $ik$-path  uncoupled  in the graph having edge matrices $\kcal_{NN}$ and $\qcal_{uu}$,  
$$ i \fla {\rm o}_a\dal  {\rm o}_b  \fla k, \nn \nn  i \fla {\rm o}_a\dal  {\rm o}_a  \fra  k,  \nn \nn  i \fra {\rm o}_b\dal  {\rm o}_b  \fla k, $$
which are closed with the induced edge matrices $\pcal_{a|b}$, $\scal_{aa|b}$,  $\scal^{bb}$, respectively, in \eqref{indpab}, \eqref{indsym}.
\end{coro}
After remembering  the types of edge at the path endpoints, we have with
$\pcal_{a|b}$ an induced  bipartite graph of arrows pointing from $b$ to $a$, with
$\scal_{aa|b}$ an induced  covariance graph, and with
$ \scal^{bb.a} $ an induced  concentration graph.

%Thereby, for an induced edge via $\pcal_{a|b}$,node $i\in a$ and $k \in b$, via $\scal_{aa|b}$ nodes $i,k\in a$, and via $\scal^{bb}$ nodes
%$i,k \in b$.

 \begin{lem}{\bf  Edge matrices induced by {\bm \Greg} for ${\bm{f_{\alpha\beta| c}}}$.} The subgraph induced
 by nodes $\alpha \cup \beta$ in  $G_{\rm reg}^{N-a|b}$ captures  the independence  implications of \Greg for   $f_{\alpha|\beta c}f_{\beta|c}$.\end{lem}

 \begin{proof}
%   and   $Y_c$,
 By the interpretation of the edge matrix components  $\pcal_{a|b},  \scal_{aa|b}, \scal^{bb.a}$, no edges are induced by taking
$$  \pcal_{\alpha|\beta.c}=[\pcal_{a|b}]_{\alpha, \beta}, \nn \nn   \scal_{\alpha\alpha |b}=[\scal_{aa|b}]_{\alpha, \alpha}, \nn \nn   \scal_{\beta \beta.a}=[\scal^{bb.a}]_{\beta \beta}.$$
Jointly,
these edge submatrices define the subgraph induced by $\alpha\cup \beta$  in  $G_{\rm reg}^{N-a|b}$.% and the independence structure  implied by \Greg for  $f_{\alpha|\beta c}f_{\beta|c}$.
\end{proof}

  The  induced graphs in node set $\alpha \cup \beta$ and   $G^{N-a|b}_{\rm reg}$  in node set $N$, are examples of  independence-predicting graphs  in contrast to   independence-preserving  graphs such as the ribbonless  graphs of Sadeghi and Lauritzen (2012)
  and  the  different types of Markov-equivalent graphs, such as     summary graphs.  With {\bf \em  independence-preserving  graphs}, one can  derive  effects of additional marginalizing and conditioning  in the starting  graph
  while {\bf \em independence-predicting graphs} can, in general,  only be used to decide on edges present or missing  in the induced graph.

  \begin{prop} {\bf Edge criteria for  implied independences and dependences.}  A dependence base \Greg   implies $\alpha\ci \beta|c$ if  $\pcal_{\alpha|\beta.c}=0$ and it implies $\alpha\pitchfork \beta|c$ if  $\pcal_{\alpha|\beta.c}\neq 0. $
 \end{prop}
 \begin{proof} The statement results with Lemma 7, equation \eqref{indpab} and Lemma 8.\end{proof}
%It is possible to show that  edge matrix criterion for independence in Proposition 4 is equivalent to the  path criterion  of the global Markov property  given in Lemma 2.\\[-3mm]
 
{\bf Distributions  satisfying all and only the  independences captured by {\bm \Greg}}. A given distribution is  said to be faithful to a graph if every of its independence constraints is captured by a given independence graph; see Spirtes, Glymour and Scheines (1993). 
For a distribution to be faithful to \Greg, it has to satisfy  the  properties needed for the graph transformations of Proposition 3, that is properties $(i)$ to $(vii)$.

 % We recall that traceable regressions  with density $f_N$ have as properties singleton transitivity, composition and intersection,
% therefore a violation of weak transitivity can only occur via path cancellations.

\begin{coro}{\bf Distributions that are faithful to {\bm\Greg}.} For a distributions with density $f_N$ generated over 
a dependence base \Greg, the following statements are equivalent\\
$(i)$ the distribution is faithful to \Greg,\\
$(ii)$ every independence and every dependence statement implied by \Greg holds for $f_N$, \\
$(iii)$ $f_N$ satisfies as additional properties: composition, intersection  and set transitivity,\\
$(iv)$ $f_N$ can be generated as a traceable regression without any path cancellations.
\end{coro}

Thus, faithfulness imposes in general an additional  strong condition on traceable sequences of  regressions.  Exceptions are,   for instance, directed acyclic graphs in which each response has only one parent.
But the most   common  situation in observational and in interventional studies is  to have  two or more regressors influencing a  response. Thus, for  using regression graphs to interpret such structures or  to plan future studies with a subset of the variables  in a subpopulation, it is not sensible to assume that a given distribution is faithful to a regression graph.  One needs  to have traceable regressions though and should investigate reasons for path cancellations if they happen to occur.\\[-9mm]
 
 \section{Discussion}
 Sequences of regressions in joint responses permit  to model changes in several response components occurring at the same time when there is an intervention. This  contrasts with interventions in sequences of regressions in single responses and in other types of chain graph models. 
  
We have 
 identified  properties of sequences of regressions in essentially arbitrary joint and single response variables and named them  traceable regressions.  A corresponding  regression graph, \Greg is a dependence base of the joint distribution in addition to capturing the independences in the regressions.
 One  knows now that the independence structure of such traceable regressions  can differ from the implications derived in terms of its generating regression graph only when there  are path  cancellations. 
 
The consequences  derivable with  a graph give changes in structure that  result in families of distributions generated over the graph while  one may not be able to generalize to this family from the structure that one can see for a distribution with one given set of parameters,   for instance as estimated in a sample.

  Sequences of traceable  regressions 
  and a regression graph \Greg  have implications for  a regression of $Y_a$ on $Y_b$ and dependences of  $Y_b$  alone when these are based on a  reordered node set $N=(a,b)$ that can be expressed with transformed  edge matrix components of \Greg. When marginalizing over $m$ in $a=\alpha \cup m$ and conditioning on $c$ in $b=\beta \cup c$,  the specific implications of \Greg
  for  the conditional densities $f_{\alpha|\beta c}$  and $f_{\beta |c}$ can now be derived  with a subgraph induced by $\alpha \cup \beta$ in this transformed graph.
   An edge matrix criterion instead of
 a path criterion gives the global Markov property of \Greg and detects,  in addition,  induced dependences when  \Greg
 is a dependence base for $f_N$.  %that is whenever there are no path cancellations.
 
 Many new questions have opened up. These include types of  conditions on a  given distribution under  which it  represents a traceable regression,  conditions on independence-predicting graphs which assure that they are  also independence-preserving, applications such as the special details needed to improve existing methods for meta-analyses,  or 
 computational aspects, such as conditions under which one type  of several   equivalent graph transformations  becomes 
 computationally much less intensive than others.\\[-3mm]

\noindent{\bf Acknowledgement.} The work reported in this paper  was undertaken during the tenure of a Senior Visiting Scientist Award by the International Agency of Research on Cancer.
I thank   G. Byrnes, D.R. Cox, G.M. Marchetti, K. Sadeghi  and the referees for their most constructive comments.\\[-9mm] 
%as well as  anonymous referees for their suggestions on how to improve the presentation.
%\renewcommand{\baselinestretch}{0.2}
\renewcommand\refname{\normalsize References.}

\end{document}

%% file: deff.tex
\usepackage[a4paper,left=28mm,right=25mm,top=30mm,bottom=30mm]{geometry}

\RequirePackage[OT1]{fontenc}
\usepackage{dcolumn,warpcol}  
\RequirePackage{amsthm,amsmath,amssymb}
\usepackage{graphicx}\graphicspath{{pics/}}

\usepackage{latexsym}
\usepackage{arydshln,booktabs,bm}
\usepackage{array}

\usepackage{here}
\usepackage{bm} % bold math

\usepackage{lscape} % page rotation

\usepackage{natbib}
\usepackage{array}

\usepackage{fancyhdr}
\usepackage{bm} % bold math
\usepackage[latin1]{inputenc}
\usepackage{color}
\usepackage{graphicx}
\RequirePackage[OT1]{fontenc}
\usepackage{latexsym}
\usepackage{booktabs,rccol,flafter,arydshln,warpcol}
\usepackage{natbib}
\RequirePackage{amsthm,amsmath, natbib}

\usepackage{here}

\newcommand{\In}{\mathrm{In}}
\newcommand{\zer}{\mathrm{zer}}

\newcommand{\txt}{\textstyle}

\newcommand{\T}{^{\mathrm{T}}}

\newcommand{\Greg}{$G^N_{\mathrm{reg}}\,$}

\newcommand\cov{\mathrm{cov}}

\newcommand{\scal}{\ensuremath{\mathcal{S}}}

\newcommand{\pcal}{\ensuremath{\mathcal{P}}}

\newcommand{\hcal}{\ensuremath{\mathcal{H}}}

\newcommand{\kcal}{\ensuremath{\mathcal{K}}}

\newcommand{\qcal}{\ensuremath{\mathcal{Q}}}
\newcommand{\wcal}{\ensuremath{\mathcal{W}}}
\newcommand{\fcal}{\ensuremath{\mathcal{F}}}

\newtheorem{prop}{Proposition}
\newtheorem{coro}{Corollary}
\newtheorem{defn}{Definition}

\newtheorem{lem}{Lemma}

\newcommand{\ci}{\mbox{\protect{ $ \perp \hspace{-2.3ex}
\perp$ }}}

\newcommand{\dep}{\pitchfork}  

\newcommand{\n}[0]{\hspace*{.35em}}

\newcommand{\nn}[0]{\hspace*{.7em}}

\newcommand{\fourl}[0]{\hspace*{1.4em}}

\newcommand{\node}{\mbox {\LARGE
{$\mbox{$\circ$}$}}}

\newcommand{\snode}{\mbox {\large
{$\mbox{$\circ$}$}}}

\newcommand{\margn}{\mbox {\raisebox{-.1 ex}{\margnn}}}
\newcommand{\margnn}{\mbox {\Large
{$\not \: \not $}}$\node $}

\newcommand{\ful}{\mbox{$\, \frac{ \nn \nn \;}{ \nn \nn
}$}}

\newcommand{\fuls}{\scriptsize{\mbox{\raisebox{-.022ex}{$\, \frac{ \nn \n \;}{ \nn \n
}$}}}}

\newcommand{\fla}{\mbox{$\hspace{.05em} \prec
\!\!\!\!\!\frac{\nn \nn}{\nn}$}}

\newcommand{\flas}{\scriptsize{\mbox{$\hspace{.05em} \prec
\!\!\!\!\!\frac{\nn \nn}{\nn}$}}}

\newcommand{\fra}{\mbox{$\hspace{.05em} \frac{\nn
\nn}{\nn
}\!\!\!\!\! \succ \! \hspace{.25ex}$}}

\newcommand{\dal}{\mbox{$  \frac{\n}{\n}
\frac{\; \,}{\;}  \frac{\n}{\n}$}}

\newcommand{\dals}{\scriptsize{\mbox{$  \frac{\n}{\n}
\frac{\; \,}{\;}  \frac{}{}$}}}

\makeatletter
\renewcommand\section{\@startsection{section}{1}{\z@}%
{-3.25ex\@plus -1ex \@minus -.2ex}{1.5ex \@plus .2ex}%
{\normalfont\large\bfseries}}
\renewcommand\subsection{\@startsection{subsection}{2}{\z@}%
{-3.25ex\@plus -1ex \@minus -.2ex}%
{1.5ex \@plus .2ex}%
{\normalfont\normalsize\bfseries}}
\renewcommand\subsubsection{\@startsection{subsubsection}{3}{\z@}%
{-3.25ex\@plus -1ex \@minus -.2ex}%
{1.5ex \@plus .2ex}%
{\normalfont\normalsize\bfseries}}
\renewcommand\paragraph{\@startsection{paragraph}{4}{\parindent}%
{3.25ex \@plus1ex \@minus .2ex}%
{-1em}%
{\normalfont\normalsize\bfseries}}
\makeatother

%% file: papDepReg27Apr12.bbl
\begin{thebibliography}{100}
{\small

\bibitem[Birch(1963)]{Birch63}
{Birch, M.W.} (1963). Maximum likelihood in three-way contingency tables.
\textit{J. Roy. Statist. Soc. B}
\textbf{25}, 220--233.
\vspace{-3mm}

\bibitem[Castillo, Guti\'errez, and Hadi(1997)]{CaGuHa97}
{Castillo, E., Guti\'errez, J. M., and Hadi, A. S.} (1997). 
\textit{Expert Systems and Probabilistic
Network Models.}  Springer,  New York.
\vspace{-3mm}





\bibitem[Cox and Wermuth(1993)]{CoxWer93}
{Cox, D.R. and Wermuth, N.} (1993).
Linear dependencies represented by
chain graphs (with discussion).
\textit{Statist. Science}
\textbf{8}, 204--218;
247--277.
\vspace{-3mm}






\bibitem[Cox and  Wermuth(1994)] {CoxWer94} 
{Cox, D.R.  and  Wermuth, N.} (1994).
Tests of linearity, multivariate
normality and adequacy of linear scores. 
\textit{J. Roy. Statist. Soc. C} \textbf{43}, 347--355.
\vspace{-3mm}



\bibitem[Darroch(1962)]{Darr62}
{Darroch, J.N.} (1962).
Interactions in multi-factor contingency tables.
 \textit{J. Roy. Statist. Soc. B}
 \textbf{24}, 251--263.
 \vspace{-3mm}

\bibitem[Dawid(1979)]{Daw79}
{Dawid, A.P.} (1979).
Conditional independence in statistical theory (with discussion). \textit{ J. Roy. Statist. Soc.  B}
\textbf{ 41}, 1--31.
\vspace{-3mm}


\bibitem[Drton(2009)]{Drton09}
{Drton, M.} (2009).
Discrete chain graph models. \textit{Bernoulli} {\bf 15}, 736--753.
\vspace{-3mm}


\bibitem[Edwards and Lauritzen(2001)]{EdwLau01}
{Edwards, D., Lauritzen, S.L.} (2001). The TM algorithm for maximising a conditional likelihood function.
\textit{Biometrika}
\textbf{88}, 961--972.
\vspace{-3mm}

\bibitem[Flesch and Lucas(2006)] {FleLuc07} 
{Flesch, I.  and Lucas, P.} (2006).  Markov equivalence in Bayesian networks.
in: \textit Advances in probabilistic graphical models.| (eds.  Lucas, P.,  G\'amez, J. and
Samer\'on, E.) Springer, Berlin,    3--38. 
\vspace{-3mm}


 \bibitem[Geiger, Verma and Pearl(1990)]{GeiVerPea90}
   {Geiger, D., Verma, T.S.  and  Pearl, J.}  (1990). Identifying independence in Bayesian networks.
   \textit{Networks}   \textbf{20},  507--534.
   \vspace{-3mm}


 
\bibitem[Hell an  Ne\v{s}et\v{r}il (2004)]{HellNes04}
{Hell, P. and Ne\v{s}et\v{r}il , J.} (2004). 
\textit{Graphs and homomorphisms.}
Oxford University Press, Oxford.
   \vspace{-3mm}
      \vspace{-6mm}

 \bibitem[Kappesser(1997)]{Kappesser97} 
 {Kappesser, J.} (1997).
 \textit{Bedeutung der Lokalisation
 f\"ur die Entwicklung und Behandlung chronischer
Schmerzen.} Thesis,  Department of Psychology,
University of Mainz.
\vspace{-3mm}






\bibitem[Lauritzen and  Wermuth(1989)]{LauWer89}
{Lauritzen, S. L. and Wermuth, N.} (1989).
Graphical models for
association between variables, some of which are qualitative and some
quantitative.
 \textit{Ann. Statist.}
 \textbf{17}, 31--57.
 \vspace{-3mm}





\bibitem [Ln\v{e}ni\v{c}ka and Mat\'u\v{s}(2007)]{LnenMatus07}
{Ln\v{e}ni\v{c}ka, R. and Mat\'u\v{s}, F.} (2007).
 On Gaussian conditional independence structures. 
 \textit{Kybernetika} \textbf{43},
323--342.
\vspace{-3mm}







\bibitem[Marchetti and Lupparelli(2011)]{MarLup11}
{Marchetti, G.M. and Lupparelli, M.}  (2011). Chain graph models of multivariate
regression type for categorical data.
\textit{Bernoulli},  \textbf{17}, 827--844.

\vspace{-3mm}



\bibitem[Marchetti and Wermuth(2009)]{MarWer09}
{Marchetti, G.M. and Wermuth, N.} (2009).
Matrix representations and independencies in
directed acyclic graphs.
\textit{Ann. Statist.}
\textbf{47}, 961--978.
\vspace{-3mm}





\bibitem[Pearl(1988)]{Pea88}
{Pearl, J.} (1988).
\textit{Probabilistic reasoning in intelligent systems.}
Morgan Kaufmann, San Mateo.
\vspace{-3mm}
   \vspace{-6mm}


\bibitem[Robins and  Wasserman(1997)] {RobWas97}
{Robins, J. and Wasserman, L.} (1997). 
Estimation of effects of sequential treatments by
reparamet\-rizing directed acyclic graphs. In: \textit{
Proc. 13th Ann. Conf., UAI} (eds. D. Geiger and O. Shenoy)
Morgan and Kaufmann,  San Mateo,  409--420.
\vspace{-3mm}
 
 



\bibitem[Sadeghi(2009)]{Sadeghi09}
{Sadeghi, K.} (2009).
Representing modified independence structures.
\textit{Transfer thesis,  Oxford University.}
\vspace{-3mm}


\bibitem[Sadeghi and Lauritzen(2012)]{SadLau12}
{Sadeghi, K. and Lauritzen, S. L. } (2012).
Markov properties of mixed  graphs.
\textit{Submitted.}
\vspace{-3mm}



\bibitem[San Martin, Mouchart and Rolin(2005)]{SanMMouRol05}
{San Martin E.,  Mochart M. and Rolin, J.M.} (2005).
Ignorable common information, null sets and Basu's first theorem.
\textit{Sankhya} \textbf{67},  674--698.
\vspace{-3mm}

	



\bibitem[Spirtes, Glymour  and Scheines(1993)]{SpiGlySch93}
 {Spirtes, P., Glymour C.  and  Scheines R.} (1993).
\textit{Causation, prediction and search}.
Springer,   New York.
\vspace{-3mm}


\bibitem[Studen\'y(2005)]{Stu05}
{Studen\'y, M.}   (2005). \emph{Probabilistic conditional independence structures}.
Springer,  London.
\vspace{-3mm}






\bibitem[Wermuth(2011)]{Wer10}
{Wermuth, N.}  (2011).
Probability models with summary graph structure.
\textit{Bernoulli},  
\textbf{17}, 845--879.
\vspace{-3mm}



\bibitem[Wermuth and Cox(1998)]{WerCox98}
{Wermuth, N. and Cox, D.R.} (1998).
On association models defined over independence graphs.
\textit{Bernoulli}
\textbf{4}, 477--495.
\vspace{-3mm}

\bibitem[Wermuth and Cox(2004)]{WerCox04}
{Wermuth, N. and Cox, D.R.} (2004).
Joint response graphs and separation induced by triangular
systems.
\textit{J.Roy. Stat. Soc. B}
\textbf{66}, 687-717.
\vspace{-3mm}



\bibitem[Wermuth and  Cox(2008)] {WerCox08}
{Wermuth, N. and Cox, D.R.} (2008).
Distortions of effects caused by indirect confounding.
\textit{Biometrika}
\textbf{95}, 17--33.
\vspace{-3mm}


\bibitem[Wermuth, Cox and Marchetti(2006)]{WerCoxMar06}
{Wermuth, N., Cox, D.R. and  Marchetti, G.M.} (2006).
Covariance chains.
\textit{Bernoulli}
 \textbf{12}, 841-862.
 \vspace{-3mm}

 

\bibitem[Wermuth and  Lauritzen(1990)]{WerLau90}
{Wermuth, N. and Lauritzen, S.L.} (1990).
 On substantive research
hypotheses, conditional independence graphs and graphical chain models
(with discusssion). \textit{J. Roy. Statist. Soc. B}
\textbf{52}, 21--75.
\vspace{-3mm}


\bibitem[Wermuth, Marchetti and Cox(2009)]{WerMarCox09}
{Wermuth, N.,  Marchetti, G.M.  and Cox, D.R.} (2009).
Triangular systems for symmetric
binary variables.
\textit{Electr. J. Statist.}
\textbf{3}, 932--955.
\vspace{-3mm}



\bibitem[Wermuth and Sadeghi (2012)]{WerSad12}
{Wermuth N. and Sadeghi, K.} (2012).
Sequences of regressions and their independences.
To appear as invited discussion paper in \textit{TEST}.
\vspace{-3mm}


\bibitem[Wermuth, Wiedenbeck and Cox(2006)]{WerWieCox06}
{Wermuth, N.,  Wiedenbeck, M. and Cox, D.R.} (2006).
Partial inversion for linear systems and partial closure of independence
graphs.
\textit{BIT, Numerical Mathematics}
\textbf{46}, 883--901.
\vspace{-3mm}





\bibitem[Wiedenbeck and Wermuth(2010)]{WieWer10}
{Wiedenbeck, M. and Wermuth, N.} (2010).
Changing parameters by partial mappings. \textit{Statistica Sinica}
\textbf{20},  823--836.
\vspace{-3mm}

  } 

\end{thebibliography}
